\documentclass[letterpaper, 10 pt, journal, twoside]{IEEEtran} 

\usepackage{amsfonts,amsmath,amssymb}
\usepackage[english]{babel}
\usepackage{hyperref}

\usepackage{amsthm}
\usepackage{graphicx}
\usepackage{subcaption,caption}
\usepackage{epsfig}
\usepackage{tikz}
\usetikzlibrary{shapes,arrows}
\usetikzlibrary{calc,patterns,angles,quotes,arrows,automata}
\graphicspath{ {./imgs/}}

\newcommand{\ket}[1]{\left| #1 \right>}
\newcommand{\bra}[1]{\left< #1 \right|}
\newcommand{\norm}[1]{\left|\left| #1 \right|\right|}

\newcommand{\Outer}[2]{ \left|#1\right>\left<#2\right| }

\newcommand{\blue}{}

\newcommand{\tr}{\text{tr}}

\newcommand{\lind}{\mathcal{L}
}

\newtheorem*{problem}{Switching control problem}

\newtheorem{proposition}{Proposition}
\newtheorem{theorem}{Theorem}

\newtheorem{definition}{Definition}
\newtheorem{assumption}{Assumption}
\theoremstyle{definition}

\usepackage[backend=bibtex, style=ieee, sorting=none]{biblatex}
\title{\LARGE \bf Stabilization via feedback switching\\ for quantum stochastic dynamics}

\author{Tommaso Grigoletto and Francesco Ticozzi
\thanks{T. Grigoletto and F. Ticozzi are with the Department of Information Engineering, University of Padova, Via Gardenigo 6, 35131 Padova, Italy. Emails: 
\href{mailto:tommaso.grigoletto@phd.unipd.it}{\texttt{tommaso.grigoletto@phd.unipd.it}},
\href{mailto:ticozzi@dei.unipd.it}{\texttt{ticozzi@dei.unipd.it}}.}}

\newcommand{\myscale}{0.55}

\begin{document}

\maketitle
\thispagestyle{empty}
\pagestyle{empty}

\begin{abstract}
We propose a new method for pure-state and subspace preparation in quantum systems, which employs the output of a continuous measurement process and switching dissipative control to improve convergence speed, as well as robustness with respect to the initial conditions. In particular, we prove that the proposed closed-loop strategy makes the desired target globally asymptotically stable both in mean and almost surely, and we show it compares favorably against a time-based and a state-based switching control law, with significant improvements in the case of faulty initialization.

\end{abstract}

\section{INTRODUCTION}

In the rapidly growing field of quantum science and technologies, the quest for new reliable and effective techniques to manipulate systems at the quantum scale is of paramount importance, and control engineers are actively contributing to this effort (see e.g. \cite{altafini-tutorial} and references therein). Due to the intrinsic nature of quantum systems, one of the central resources of classical control design - {\em feedback} control - is particularly difficult to harness. In the last two decades, great progress has been made in this sense, building on the foundation of quantum probability and filtering theory \cite{belavkin,quantum-filtering,ticozzi2012stabilization}, and arriving to remarkable experimental implementations \cite{haroche-feedback}. 

In this paper, we combine the advantages offered by measurement-based feedback with switching control strategies, which allow us to include {\em dissipative} control resources in a systematic way, as opposed to the more typical Hamiltonian control (see e.g. \cite{VanHandel2005, Mirrahimi2007,} and \cite{altafini-tutorial} for a review). {\blue The main contribution is a switching control strategy for the stabilization of a target pure state or subspace, where the current dynamics is selected based on the estimation of the state at the switching times, and maintained for a given dwell time. Such strategy guarantees practical stability of the target in mean under minimal assumptions (existence of a switching control Lyapunov function), while under typical control assumptions the stability is guaranteed both in mean and almost surely.}

The paper is structured as follows: Section \ref{sec:prob_def} defines the problem of interest and present some theoretical results that are used in the rest of the paper; {\blue Section \ref{sec:control_strategies} discusses the assumptions under which effective strategies can be derived, recalls open-loop switching strategies previously introduced in \cite{Scaramuzza-Ticozzi}, presents the novel measurement-based, closed-loop control laws and proves its stability. The analysis of convergence uses techniques that depart from those based on linear systems used in \cite{Scaramuzza-Ticozzi} and build on specific properties of the stochastic filtering dynamics \cite{Ticozzi2017}.} Section \ref{sec:case_study} provides some insights on the performances of the novel control law by presenting a relevant case study and simulation results for the stabilization of entangled states on networks of two-level systems.

\section{PROBLEM DEFINITION}
\label{sec:prob_def}
In this article we consider a system described by a finite-dimensional Hilbert space $\mathcal{H}$. Following the standard in physics and quantum information science, we shall employ Dirac's notation for vectors $\ket{\psi}\in\cal{H},$ and their duals $\bra{\psi}\in\cal{ H}^\dag.$ Let $\mathcal{B}(\mathcal{H})$ denote the set of linear operators on $\mathcal{H}$.
The state of a quantum system is completely described by a density operator $\rho\in\mathfrak{D}(\mathcal{H}) = \{\rho\in \mathcal{B}(\mathcal{H}):\rho=\rho^\dag\geq0, \ \tr(\rho)=1\}$. 

We will suppose that the system is controlled with a series of driving dynamics and that it is subjected to an homodyne detection measurement. The resulting dynamics is thus described by processes $(\rho_t)_{t\in\mathbb{R}_+}$ of states associated to the stochastic master equation (SME) \cite{quantum-filtering,altafini-tutorial}
\begin{equation}
    d\rho_t=\mathcal{L}_j(\rho_t)dt+\mathcal{G}_C(\rho_t)dW_t,
    \label{eqn:model}
\end{equation}
where: 
\begin{align*}
    \mathcal{L}_j(\rho_t) &= -i[H_j,\rho_t] + \mathcal{D}_{L_j}(\rho_t) +\mathcal{D}_C(\rho_t),\\
    \mathcal{D}_A(\rho_t) &= A\rho_t A^\dag-\frac{1}{2}\{A^\dag A, \rho_t\},\\
    \mathcal{G}_C(\rho_t) &= C\rho_t +\rho_t C^\dag -\tr\left((C+C^\dag)\rho_t\right)\rho_t
\end{align*}
In equation \eqref{eqn:model}, the term $-i[H_j,\rho_t] + \mathcal{D}_{L_j}(\rho_t)$ represents the driving dynamics associated to the Hamiltonian and noise operators $(H_j, L_j)$. We suppose that a list of possible driving dynamics is provided and the role of the controller will be to choose which dynamics to activate. This concept will be further formalized in the following. 
The term $\mathcal{D}_C(\rho_t)dt+\mathcal{G}_C(\rho_t)dW_t$ accounts for the homodyne detection measurement process associated to the fixed noise operator $C$. For simplicity we here consider unit detection efficiency, but the control strategy works with imperfect detection as well. The process $dW_t$ is a Wiener process, adapted to the filtration ${\cal F}_t$ \cite{quantum-filtering}, and it can be seen as the innovation for the %$dW_t\sim\mathcal{N}(0,dt)$ 
homodyne measurement output $ dY_t =\tr\left((C^\dag+C)\rho_t\right)dt+dW_t.$

If the measurement record is not accessible, the best description of the state evolution can be obtained as the expectation of \eqref{eqn:model} over the outcomes of the measurement process. Namely, defining $\hat{\rho}_t = \mathbb{E}[\rho_t|{\cal F}_t]$, we have that the time evolution of $\hat{\rho}_t$ is described by the Markovian master equation (MME) \cite{alicki-lendi}:
\begin{equation}
        \frac{d}{dt}\hat{\rho}_t = \mathcal{L}_j(\hat{\rho}_t).
        \label{eqn:MME_model}
\end{equation}

Throughout this work, we will consider the stabilization of linear subspaces of $\mathcal{H}$, and the relevant particular case of pure states. 
Let $\mathcal{H_S}$ be the target subspace of $\mathcal{H}$. Denoting $P_S$ the orthogonal projector on $\mathcal{H_S}$, we can describe the set of states whose support is $\mathcal{H_S}$ or a subspace of $\mathcal{H_S}$ as 
\begin{equation}
    \mathcal{I_S(H)} = \{\rho\in\mathfrak{D}(\mathcal{H}): \tr(P_S\rho)=1\}.
\end{equation}
With a slight abuse of terminology, we say that $\mathcal{H_S}$ is invariant, stable, or attractive for the dynamics if such is the supported state-set $\mathcal{I_S(H)}.$ 
    The subspace $\mathcal{H_S}$ is said globally asymptotically stable (GAS) for \eqref{eqn:model}:\vspace{1mm}\\ - in mean if  $\lim\limits_{t\to\infty}\norm{\hat{\rho}_t - P_S \hat{\rho}_t P_S }_1=0$, $\forall\rho_0\in\mathfrak{D}(\mathcal{H})$;\\
    - almost surely ({\em  a.s.}) if $\mathbb{P}\left(\lim\limits_{t\to\infty}\norm{\rho_t - P_S \rho_t P_S }_1=0\right)=1$, $\forall\rho_0\in\mathfrak{D}(\mathcal{H})$.

    We next summarize some results of \cite{Ticozzi2017} that will be of use to our aims in the following theorem. {\blue A Lyapunov function is a functional $V:\mathfrak{D}(\mathcal{H})\to\mathbb{R}_+$ such that: $V(\rho)\geq0$, with $V(\rho)=0$ if and only if $\rho\in\mathcal{I_S(H)}$ and $\frac{d}{dt}V(\rho)<0$ for all $\rho\notin\mathcal{I_S(H)}$.

    \begin{theorem}
        Consider system \eqref{eqn:model}, with a fixed $\lind_j$. A subspace $\mathcal{H_S}$ of $\mathcal{H}$ is:
        \begin{itemize}
            \item GAS in mean if and only if it is GAS almost surely;
            \item if GAS in mean if and only if there exists an operator $K\geq0$ such that $V(\rho)=\tr(K\rho)$ is a Lyapunov function. 
        \end{itemize}
        \label{thm:mean_as_equ}
    \end{theorem}
    {\em Remark:} These results are the starting point for our analysis and deserve some explanations: (i) The linear  $V$ is associated with a $K$ such that $K_R:=P_SKP_S>0,$ and the latter matrix is derived from a (perturbed) Perron-Frobenius eigen-operator for the dual dynamics, reduced to the complement of the target. (ii) The equivalence of asymptotic stability in mean and a.s. is {\em crucially dependent on the fact the target is a subspace}, and not true otherwise. In our main theorem we will use the same proof idea, applied to switching evolutions.}

We shall be interested in sequences of switching times that are unbounded countable set of times $0,t_1,t_2,\dots$ such that $t_k-t_{k-1}>\epsilon$ for some $\epsilon>0.$  The assumption $t_k-t_{k-1}>\epsilon$ for $\epsilon>0$ is introduced to prevent chattering: for this reason, we shall refer to a sequence with the properties above as a {\em non-chattering}. {\blue Non-chattering time sequences are essential for practical implementations, as well as ensuring well-behaved solution of the SME. This point will not be explicitly discussed in this work but the proof follows the discussion of \cite{Mirrahimi2007}.} Finally, we state the control problem of interest.

\begin{problem}
    Given a target subspace $\mathcal{H_S}$ of $\mathcal{H}$ and a finite set of generators $\{\mathcal{L}_j\}_{j=1,\dots,m}$ for model \eqref{eqn:model}, find a piece-wise constant switching control law $j(t):[0,+\infty)\to\{1,\dots,m\}$ that admits a set of non-chattering switching times, so that $\mathcal{H_S}$ is made GAS (in mean and/or almost surely) by selecting  $\mathcal{L}_{j(t_k)}$ on $[t_k,t_{k+1})$. 
\end{problem}

\section{CONTROL STRATEGIES}
\label{sec:control_strategies}
In this section, we will first recall two previously proposed switching techniques, based on {\em open-loop} switching, and next present our {\em closed-loop} proposal.

\subsection{Control Assumptions}
The following assumption is typically required to prove convergence of a switching law. 
    \begin{assumption}
        \label{assumption}
            {Each Lindblad generator has the target subspace $\mathcal{H_S}$ as invariant , and}
            there exists $\alpha\in[0,1]^m$, $\norm{\alpha}_1=1$ such that $\mathcal{H_S}$ is GAS in mean for
            \begin{equation}
               \label{eq:assump1} \frac{d}{dt}\hat{\rho}_t = \mathcal{L}_{c}(\hat{\rho}_t) = \sum_{j=1}^m\alpha_j \mathcal{L}_{j}(\hat{\rho}_t).
            \end{equation} 
    \end{assumption}
     We also introduce a second working assumption, which relaxes Assumption 1 above. Essentially, it requires a linear control Lyapunov function. 
    \begin{assumption}
    \label{assumption_2}
        There exists a linear Lyapunov function $V(\rho)$ such that  $\forall \rho\notin\mathcal{I_S(H)}$, $\exists j$: $V(\mathcal{L}_j(\rho))<0$ and $\forall \rho\in\mathcal{I_S(H)}$, $\exists j$: $V(\mathcal{L}_j(\rho))=0$.
    \end{assumption}
    One of the key differences between the assumptions is that the second does not require invariance of the target for each generator, and is thus weaker, as we argue in the following.
    \begin{proposition}\label{prop:1implies2}
        Assumption \ref{assumption} implies Assumption \ref{assumption_2}.
    \end{proposition}
    \begin{proof}
        If Assumption 1 holds, from Theorem \ref{thm:mean_as_equ} we have that there exists a linear Lyapunov function $V(\rho)$ such that $V(\rho)>0$ for all $\rho\notin\mathcal{I_S}(\mathcal{H})$, $V(\rho)=0$ for all $\rho\in\mathcal{I_S}(\mathcal{H})$ and $V(\mathcal{L}_c(\rho))<0$ for all $\rho\notin\mathcal{I_S}(\mathcal{H})$. This means that we have 
       $ V(\mathcal{L}_c(\rho))=\sum_{j=1}^m\alpha_j V(\mathcal{L}_{j}(\rho))<0 
      $
        for all $\rho\notin\mathcal{I_S}(\mathcal{H})$, thanks to the linearity of $V(\rho)$. Then, since $\alpha_i>0$ and $\norm{\alpha}_1=1$, we have that there exist an index $
            j^* = \arg\min_{j=1,\dots,m} V(\mathcal{L}_{j}(\rho))
      $
        such that
       $ V(\mathcal{L}_{j^*}(\rho)) \leq \sum_{j=1}^m\alpha_j V(\mathcal{L}_{j}(\rho)) = V(\mathcal{L}_c(\rho))< 0
       $
        for any $\rho\notin\mathcal{I_S}(\mathcal{H})$.
        {Finally, since from Assumption \ref{assumption} we have that every Lindblad generator leaves $\mathcal{H_S}$ invariant, we have that $\forall j$ $V(\mathcal{L}_j(\rho))=0$,  $\forall\rho\in\mathcal{I_S(H)}$.}
    \end{proof}
   Proposition \ref{prop:1implies2} will be instrumental to the proof of the main theorem.  It is possible to prove that the converse implication is not true. For example, consider the following three-level system in ${\cal H}=\textrm{span}\{\ket{0},\ket{1},\ket{2}\}$. Let $\ket{0}\bra{0}$ be the target state, and consider the two following generators:
    \begin{align*}
        \mathcal{L}_1(\rho):( &H_1 = 0, L_1=\Outer{0}{2}+\Outer{1}{1}+\Outer{2}{2}),\\
        \mathcal{L}_2(\rho):( &H_2=\Outer{0}{2}+\Outer{2}{0},\\
		&\quad L_2=\Outer{0}{1} +\Outer{1}{1}+\Outer{2}{2} ).
    \end{align*}
    Notice that the generator $\mathcal{L}_1(\rho)$ stabilizes the space generated by $\textrm{span}\{\ket{0},\ket{1}\},$ being the only invariant space for $L_1$  and has trivial dynamics on it \cite{ticozzi-QDS}, but does not make the target GAS. For the generator $\mathcal{L}_2(\rho)$, instead, $\ket{0}$ is not even invariant.
    It is then quite clear that there exists no convex combination of $\mathcal{L}_1(\rho)$ and $\mathcal{L}_2(\rho)$ that makes $\ket{0}\bra{0}$ GAS in mean.
    However, by considering $V(\rho) = \tr(K\rho)$ with $K$ the projector on the subspace orthogonal to the target,
	it is possible to prove by direct computation that $\forall\rho\neq\ket{0}\bra{0}$, $V(\mathcal{L}_j(\rho))<0$ for some $j$ and $V(\mathcal{L}_j(\ket{0}\bra{0}))=0$ for $j=1,2$. 
	
    This shows that Assumption \ref{assumption_2} is genuinely more general than Assumption \ref{assumption}. However, if one requires non-chattering control strategies, we will argue it only allows for a weaker stability notion, and convergence in mean.
    
    %%%%
    \subsection{Open-loop strategies: cyclic and state-based switching}
    
    The first control strategy we consider is a {\em time-based} solution \cite{Scaramuzza-Ticozzi}, based on Assumption \ref{assumption}.
    \begin{definition}[Cyclic switching control law]
        Given the vector $\alpha$ that satisfies Assumption \ref{assumption} for the set of Lindblad dynamics $\{\mathcal{L}_j\}$, the cyclic switching control law selects each index $j$ for a fraction $\alpha_j$ of the total cycle period $\varepsilon>0$.
    \end{definition}
    
    Note that this control law depends only on the vector $\alpha$ from Assumption \ref{assumption}, and no information on the initial state of the system is required. Essentially, for $\varepsilon\rightarrow 0$ it mimics the evolution generated by the convex combination \eqref{eq:assump1}, and makes $\mathcal{H_S}$ GAS in mean. A full proof can be found in \cite{Scaramuzza-Ticozzi}.
    The second control law we recall is  also proposed in \cite{Scaramuzza-Ticozzi}, and here specialized to the case of a linear Lyapunov function.
    
    \begin{definition}[State-based switching control law]
        Given a set of Lindblad operators $\{\mathcal{L}_j\}$, an estimate of the initial state $\rho_0$, a linear Lyapunov function $V(\cdot)$ that guarantees that $\mathcal{H_S}$ is GAS for the system $\mathcal{L}_c$ as in \eqref{eq:assump1} and a {non-chattering} switching time sequence $\{t_k\}$, the state-based switching control law is defined as:
        \begin{equation}
           j(t) = \arg\min_{j=1,\dots,m} V(\mathcal{L}_j(\hat{\rho}_{t_k})) , \quad \forall t\in[t_k,t_{k+1})
        \end{equation}
        where $\hat{\rho}_{t_k}$ is the average of the system state at time $t_k$, solution of the MME.
    \end{definition} 
    The original proposal in \cite{Scaramuzza-Ticozzi} employed a quadratic function to show the existence of a non-chattering sequence; in this case, it follows from Theorem \ref{thm:main} in the following section.
    
    \subsection{Measurement-based switching control}
    
    The approach we propose exploits continuous measurements to have a {\em closed-loop} estimate of the current state of the system, and then apply the state-based switching control law based on this estimate - as opposed to the MME average state.

    An important step in setting up the feedback loop is the choice of the measurement operator $C$, so that the measurement process does not destabilize the target  - namely, the target should be an eigenstate or, more generally, an eigenspace, of $C$.
    A natural choice is to resort again to Theorem \ref{thm:mean_as_equ} and choose $C=K$, where $K$ is the matrix of the Lyapunov function $V(\rho)=\tr(K\rho)$ that satisfies Assumption \ref{assumption_2} for the set of MME generators $\{\mathcal{L}_j\}$. In this way we have that $\mathcal{H_S}\subseteq \textrm{ker}(C)$ is naturally satisfied by the construction of $K$, and also it holds that $V(\mathcal{D}_K(\rho))=0$ for all $\rho$. From now on we will only consider this choice of $C$ for sake of simplicity.
    The switching control law we consider is:
    \begin{definition}[Measurement-based switching control law]
    \label{def:measurement_based}
        Given a set of Lindblad operators $\{\mathcal{L}_j\}$, an estimate of the initial state $\rho_0$, a $V(\rho) = \tr(K\rho)$ that satisfies Assumption \ref{assumption_2} for the set $\{\mathcal{L}_j\}$ 
        and a {non-chattering} switching time sequence $\{t_k\}$, the measurement-based switching control law is defined as:
        \begin{equation}
            j(t) = \arg\min_{j=1,\dots,m} V(\mathcal{L}_j(\rho_{t_k})), \quad \forall t\in[t_k,t_{k+1})
        \end{equation}
        where $\rho_{t_k}$ is the estimate of the state at  $t_k$ obtained as the solution of equation \eqref{eqn:model}.
    \end{definition}

    We will now prove stability of the proposed switching law. 
    
    \begin{theorem}\label{thm:main}
    {\blue If Assumption \ref{assumption} holds, then there exists a non-chattering time sequence such that} the measurement-based switching control law makes $\mathcal{I_S}(\mathcal{H})$ GAS in mean and almost surely.
    \end{theorem} 
    
    \begin{proof}
        {\blue
        We start by proving the existence of a non-chattering time sequence, then we prove stability in mean, and finally, we argue that a.s. stability follows  from stability in mean and Theorem \ref{thm:mean_as_equ}.
        If Assumption \ref{assumption} holds, then by Proposition \ref{prop:1implies2} we can construct a linear Lyapunov $V(\rho)=\tr(K\rho)$ function as in Assumption \ref{assumption_2}.
        Assume $\dot{V}(\rho_{t_k}) := V(\mathcal{L}_j(\rho_{t_k}))<0$ and define $\hat{t} = \min\{t \text{ s.t. }V(\rho_t) \leq r V(\rho_{t_{k}})\}$ for some $r\in[0,1)$ fixed. We want to prove that $\exists\varepsilon>0$ such that $\hat{t}-t_k>\varepsilon$, and that it does not depend on the state at time $t_k$.
        
        We can then apply the mean value theorem: there exists $s\in[t_k,\hat{t}]$ such that $\ddot{V}(\rho_s)(\hat t-t_k) = \dot{V}(\rho_{\hat t}) - \dot{V}(\rho_{t_k}) = (1-r)|\dot{V}(\rho_{t_k})|$.
         Since by Assumption \ref{assumption} all evolutions leave $\mathcal{I_S}(\mathcal{H})$ invariant we have that:
        \begin{align}
            | \dot V(\rho_{t_k}) | & = |\tr( \lind_{j}^\dag(K) \rho_{t_k})| \geq |\tr( \lind_C^\dag(K) \rho_{t_k})|\nonumber \\ 
            &= |\tr( \lind_{C,R}^\dag(K_R) \rho_{R,t_k})|\geq k_{C,min} \tr(\rho_{R,t_k}), \label{eq:bound}
         \end{align}
        where the $R$ subscript denotes the restriction of the operators to the complement of the target space and $k_{C,min}$ is the minimum absolute value of the eigenvalues of $\lind_{C,R}^\dag(K_R),$ and in the second line we used  Proposition 2.5 of \cite{Ticozzi2017}. 
        On the left, we then have
        $| \ddot V(\rho_s) | = | \tr( K \lind^2_j(\rho_S)|= |\tr( K_R \lind^2_{R,j}(\rho_{R,s})| \leq \bar k_{2,max} \tr(\rho_{R,s}),$
        where $\bar k_{2,max}$ is the maximum absolute value of the eigenvalues of $(\lind_{R,j}^{\dag})^2(K_R),$ for all $j$. 
        Then, since $e^{(t-t_j)\lind_{R,j}}$ is a trace-non-increasing map for any finite interval and $\rho_{R,t}=e^{(t-t_k)\lind_{R,j}}\rho_{R,t_k}$, we have $\tr(\rho_{R,t_k})\geq\tr(\rho_{R,t})$, $\forall t \in[t_k,\hat{t}]$.
        Combining all the above inequalities we obtain:
        $ \hat t-t_k\geq {(1-r)k_{C,min}}/{\bar k_{2,max}} $
        then taking $t_{k+1}\in[t_k,\hat{t}]$ we obtain a non-chattering time sequence such that $\dot{V}(\rho_t)<0$ $\forall t \in[t_k,t_{k+1}]$.}
        We will now prove stability in mean. We shall use  a generalized version of Barbalat's Lemma and in particular Corollary 1 in \cite{Huang2012} to prove stability. 
        We have that $V(\rho)$ is lower bounded since $V(\rho)\geq0.$ By construction, $\dot{V}(\rho) =V(\mathcal{L}(\rho))$ is piece-wise continuous, and
     non-positive in the time interval for a non-chattering sequence constructed as above. Moreover, being $V(\mathcal{L}(\rho))$ linear in $\rho$ for any time interval $t\in[t_i,t_{i+1})$, we have that $V(\rho)$ is twice differentiable in any time interval. Computing the second derivative of $V(\cdot)$ with respect to time, we get $\ddot{V}(\rho)=\tr(\mathcal{L}_j^\dag(K)\mathcal{L}_j(\rho))$ which is linear in $\rho$, hence bounded for all $t\in[t_i,t_{i+1})$. Thus all hypothesis of Corollary 1 in \cite{Huang2012} are satisfied, and $\dot{V}(\rho)\to0$ in mean for $t\to\infty$.
     {\blue This proves stability in mean, and that $V(\rho_t)$ is a (continuous) positive supermartingale for SME. Convergence in mean implies $L^1$ convergence to 0; on the other hand, by bounded supermartingale convergence theorem $V(\rho)$ converges both in $L^1$ and a.s. to some $V_\infty$ (see also proof of Theorem 1.1 in \cite{Ticozzi2017}.) Since we know that the $L^1$-limit is $0,$ convergence a.s. is also guaranteed. Being $V$ positive on the complement of the target, $V(\rho)=0$ a.s. implies $\rho\in{\cal I}_S({\cal H})$ a.s. as well.}
    \end{proof}
    
    {\blue 
    {\em Remark:} For $\Delta t=t_{k+1}-t_k\rightarrow 0$ the proposed strategy tends to select {\em each time} $\lind_{j(t)}$ such that
    $j(t)=\textrm{argmin}_jV(\lind_j\rho_t)$. While granting the {\em optimal} convergence rate at each time, at least if the estimated state is correct, this continuous strategy is not practically viable, and could lead to chattering. Assumption \ref{assumption}, however, allows us to derive non-chattering switching sequences (and a worst-case exponential bound using the dominant eigenvalue of $\lind_{C,R}$). Hence it is key to assess how well the proposed strategy fares with respect to the optimal one with faulty initializations. In the next sections we shall focus on this, leaving the estimate of the optimal convergence rate for future work.
    We next show that the weaker Assumption \ref{assumption_2} is still sufficient to prove {\em practical} stability in mean.
    \begin{theorem}
    Assume Assumption 2 to hold for $V(\rho)=\tr(K\rho)$. Then for every $\varepsilon>0$ there exists a non-chattering sequence such that the measurement-based strategy stabilizes the $\varepsilon$-``neighborhood" of the target ${\cal I}_{S,\varepsilon}=\{\ \rho\in  \mathfrak{D}(\mathcal{H})|\min_{\sigma\in{\cal I}_S} \|\rho - \sigma\|_1<\varepsilon\}$ in mean, and enters it in finite time. 
    \end{theorem}
    \begin{proof}
    First notice that
    $\min_{\sigma\in{\cal I}_S} \|\rho - \sigma\|_1=\min_{\sigma\in{\cal I}_S} \tr(|\Pi_S(\rho - \sigma)|)+\tr(\rho_R)\leq 2\tr(\rho_R),$ with $\Pi_S$ the projector on the support of the target.
    So  $\tr(\rho_R)\leq \varepsilon/2$ implies $\rho\in{\cal I}_{S,\varepsilon},$ and, since $V(\rho)<k_{max}\tr(\rho_R),$ we have $\Omega_{\varepsilon/2k_{max}}\subseteq{\cal I}_{S,\varepsilon},$ where we define the sub-level set  $\Omega_\ell=\{\rho|V(\rho)\leq \ell\}.$ So it is sufficient to prove that $\Omega_\beta=\Omega_{\varepsilon/2k_{max}}$ is stabilized by the average dynamics.
    Consider $\alpha<\beta=\varepsilon/2k_{max}.$ 
    By Assumption 2,  and since $ \mathfrak{D}(\mathcal{H})$ is compact and $\dot V$ and $\ddot V$ are linear in $\rho$, there exists $\delta>0$ such that $\min_j V(\lind_j\rho)<-\delta$ for all $\rho\notin \Omega_\alpha,$ as well as two constants such that  $\max_{\rho,j}|V(\lind_j\rho)|<M_1$ and $\max_{\rho,j}| V(\lind_j^2\rho)|<M_2.$ Then, if we take $t_{k+1}-t_k=\Delta t< \delta/M_2$ we are guaranteed that $\dot V$ remains negative for the whole switching interval. Since $\dot V(\rho)<-\delta+M_2\Delta t<0,\;\forall\rho\notin \Omega_\alpha,$ we have  $V(t)\leq V(t_k)+(-\delta+M_2\Delta t)(t-t_k),$ and $\Omega_\alpha$ is reached in finite time during some interval $[t_\ell,t_{\ell+1})$, and so is $\Omega_\beta,$ and remain in there until at least $t_{\ell+1}.$  We next prove that if we start from $\Omega_\alpha$ we do not exit $\Omega_\beta,$ making it invariant. To this aim, it is sufficient to note that if $\rho_{t_{\ell+1}}\in\Omega_\alpha$, then for $t>t_{\ell+1}$:
    $V(\rho_t)\leq \alpha + M_1\Delta t$. The right-hand side can be made arbitrarily small by reducing $\alpha$ and $\Delta t$, so it is always possible to ensure $\alpha + M_1\Delta t\leq \beta.$
    \end{proof}
    {\em Remark:} In this case, the lack of invariance of the target prevents from guaranteeing non-zero dwell times close to it; as a consequence, stability in mean is proved towards a set that does not have limited support, hence convergence a.s. cannot be established as before.}
    
\subsection{On robustness with respect to initialization errors}
\label{sec:initialization}

The knowledge of the initial condition plays a central role in determining the effectiveness of the control law, as already highlighted above.
 Let suppose that the system has a true initial state $\rho_0$ but our best estimate on the initial state is $\rho_{0,e}$. 
If we aim to stabilize a pure state, then by a simple majorization argument (see \cite{Scaramuzza-Ticozzi}) it is easy to show that convergence for the state-based strategy is guaranteed for any initial condition such that $supp(\rho_0)\subseteq supp(\rho_{0,e}),$ even if the switching generators are selected using the projected evolution of the wrong state. However, it is not straightforward to estimate how fast convergence is attained, or what would be the worst-case scenario.

When considering the SME, in the case of a non-exact estimate of the initial condition, the evolution of the true state is still described by model \eqref{eqn:model}. Then, by using the fact that the output signal $Y_t$ is independent of how we model the state, we obtain the evolution of the estimated state:
\begin{align}
    d\rho_{t,e} &= \mathcal{L}_j(\rho_{t,e})dt + \mathcal{G}_C(\rho_{t,e})d\widetilde{W}_t,\\
    d\widetilde{W}_t &=dW_t + \tr[(C+C^\dag)(\rho_t - \rho_{t,e})]dt.
    \label{eqn:filter_est}
\end{align}
It is important to highlight that, while $dW_t$ is a Wiener process, $d\widetilde{W}_t$ is not, as it presents a drift term. This complicates the study of the stability and convergence of the filter: some useful results are provided in \cite{VanHandel2009, Amini2011}.
{\blue In particular, for homodyne-detection SMEs as in our case, it is known that the filter is stable, and the measurement outcomes estimate converge, namely $\tr[(C+C^\dag)(\rho_t - \rho_{t,e})]\rightarrow 0$ in mean. For extremal eigenvalues of $C$, as in our case, this directly implies that {\em attaining stabilization for the estimate also implies stabilization of the actual state, and vice-versa.} If the measurement effect is small with respect to the drift (e.g. we consider $\lambda C$ for a positive, sufficiently small $\lambda$) then the SME behavior is well approximated by the MME, and robustness is inherited whenever $supp(\rho_0)\subseteq supp(\rho_{0,e})$.
General convergence of the density operator is harder to obtain, and may not always be granted. When Assumption \ref{assumption} holds, the target is invariant and $\tr(\hat\rho_{R,t})=\tr(P_R\hat\rho_t )$ is decreasing, for both the actual and the estimated state. This implies that the SME dynamics are stable and the two evolutions do not diverge in mean. 

It goes beyond the scope of this work to analyze what are the conditions under which the switching law stabilizes both. we tested initialization robustness numerically finding good results in all simulations. What is affected by faulty initialization, in general, is the convergence speed and we will exhibit some evidence in the next section.
Being the measurement-based control law a true {\em closed-loop} solution, we expect it to have a faster convergence than its open-loop counterpart, with the state estimate getting more accurate using the measurements and converging to the true state faster. }

\section{Application to multipartite entanglement generation: graph states}
\label{sec:case_study}
	\begin{figure*}[h]\centering
    	\begin{subfigure}{.49\textwidth}\centering
    		\includegraphics[scale=\myscale]{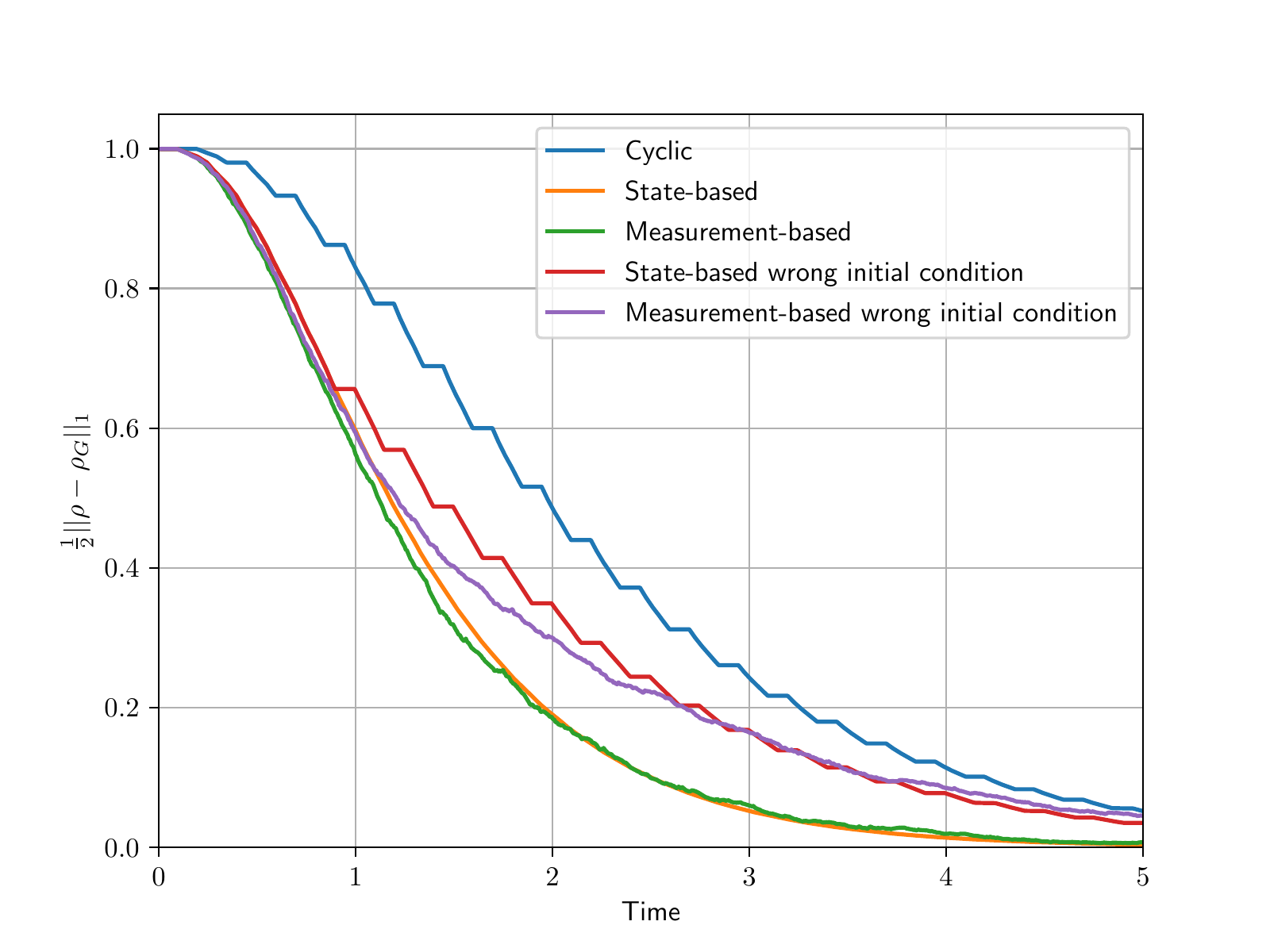}
    		\caption{Results of simulation 1.}
    		\label{fig:results_1}
    	\end{subfigure}
    	\begin{subfigure}{.49\textwidth}\centering
    		\includegraphics[scale=\myscale]{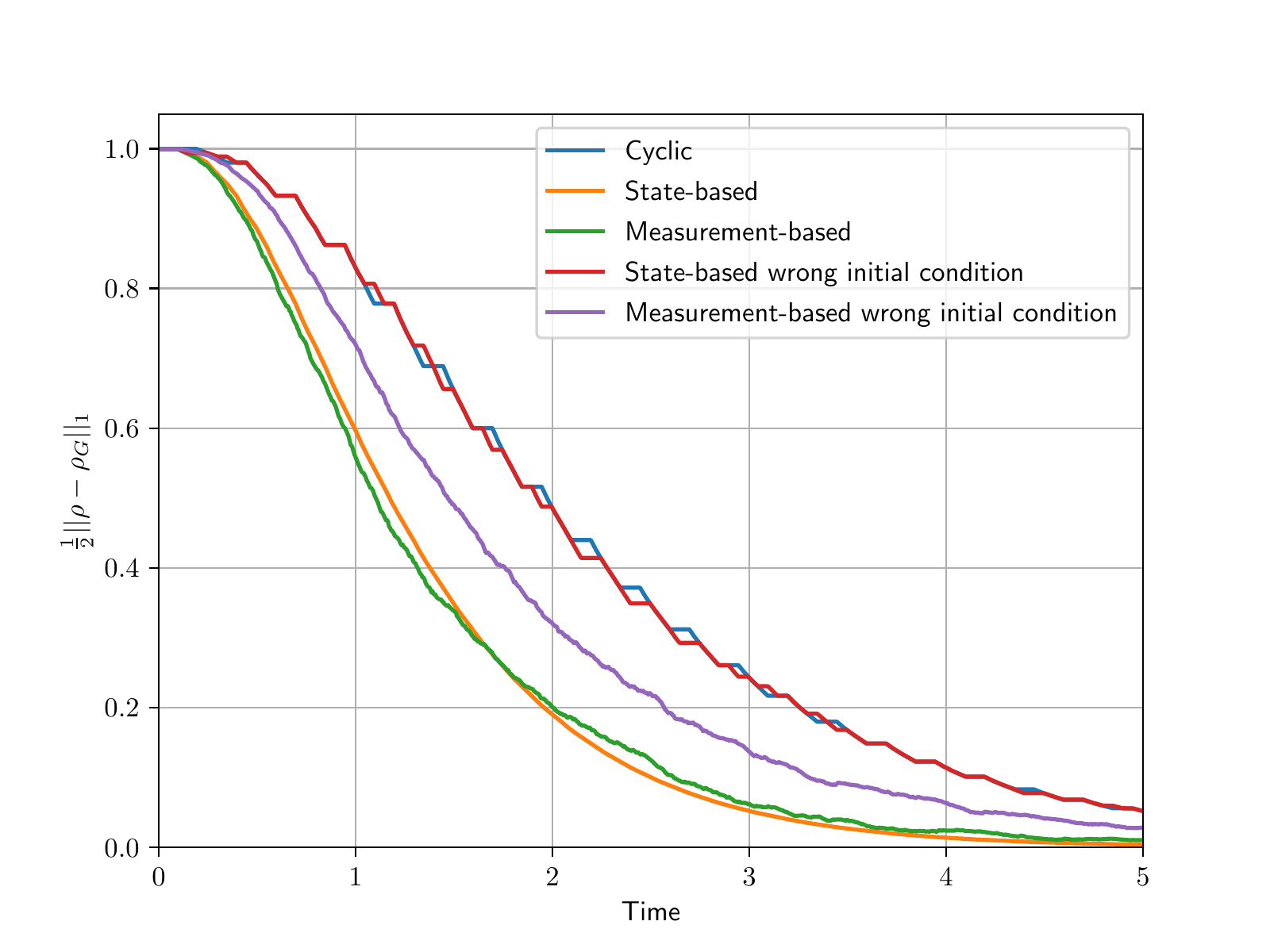}
    		\caption{Results of simulation 2.}
    		\label{fig:results_2}
    	\end{subfigure}
    	\caption{Results of the two simulations scenarios described in paragraphs \ref{par:sim_1} and \ref{par:sim_2}. Only the trajectories of the true state of the system are here reported (not the estimated one).} 
    	\label{fig:results}
    \end{figure*}
    In this section, we test the performance of the proposed control law in stabilizing an entangled state of a network of qubits (two-dimensional systems). We focus on the case of graph states: these are of practical interest \cite{graph} and offer a highly symmetric structure on which we can rely upon to design our control, which we briefly recall in the following.
	Let us consider a graph $\mathcal{G} = (\mathcal{V},\mathcal{E})$ with $|\mathcal{V}|=n$. Each node of the graph $j\in\mathcal{V}$ is associated to a qubit:
	$
	\mathcal{H} = \bigotimes_{a=1}^n\mathcal{H}_a, \quad \textrm{dim}(\mathcal{H}_a)=2\quad\forall a.
	$
	Each edge $(i,j)\in\mathcal{E}$ describes the interaction between two adjacent qubits, this interaction is described by the operator $U_{(j,k)}$. We will assume that $U_{(j,k)}$ acts as the generalized controlled-Z gate acting on the qubits $j$ and $k$. Thus we have $U_{(j,k)}=C^Z_{(j,k)}\otimes I_{\overline{(j,k)}}$ with $C^Z = \text{diag}(1,1,1,-1)$.
	The product of all such unitary matrices, $	U_G = \prod_{(j,k)\in\mathcal{E}}U_{(j,k)}$
	can be seen as a global unitary, or in quantum computation terms a quantum circuit, which is used to map a factorized states with respect to the original qubit subsystems into  entangled states, called {\em graph states}. A more comprehensive description of these states and their open-loop stabilization can be found in \cite{Ticozzi2012,Johnson2016}.
	
	We consider a particular pure state as our target state but the procedure that follows can be adapted to any pure target state. Assume that we want to prepare the state $\rho_G = U_G \ket{+}\bra{+}^{\otimes n} U_G^\dag$.
	To construct (locality constrained) MME generators that stabilize $\rho_G$ we simply construct the local generators which prepare each factor into $\rho_j=\ket{+}\bra{+}$ in the un-rotated basis, extend each generator to the whole system by tensor product with the identity and then transform these operators by $U_G$. The operators needed for local stabilization are thus of the form 
	$L_{G,j} = U_G^\dagger (\ket{+}\bra{-}_j\otimes I_{\bar{j}})U_G.$
	We can then create a  $\mathcal{D}_{L_j}(\rho)$ for each qubit, using these noise operators.
	In \cite{Johnson2016} it is possible to find the proof of the fact that 
	$\mathcal{D}_c(\rho) = \sum_{j=1}^n\frac{1}{n}\mathcal{D}_{L_j}(\rho)$
	makes $\rho_G$ GAS for the model $\frac{d}{dt}\rho_t = \mathcal{D}_c(\rho_t)$. 
	In order to  find a Lyapunov function suitable to our aims we introduce the graph Hamiltonian:
    $ H_G =-U_G^\dagger\Big(\sum_jX_j\Big)U_G$
	where $X_j=\sigma_{x,j}\otimes I_{\bar{j}}$ and $\sigma_{x,j}$ represents the Pauli-$x$ matrix acting  on the $j$-th qubit. 
	The target state $\rho_G$ may be seen as the unique ground state of the graph Hamiltonian $H_G$, i.e. $\rho_G$ is the only state such that $\tr(H_G\rho_G)=\min\lambda_i$, where $\lambda_i$ are the eigenvalues of $H_G$. 
	This implies that $V(\rho) = \tr(H_G\rho)$ is a valid linear Lyapunov function. 
    We can thus consider models \eqref{eqn:model} and \eqref{eqn:MME_model} with $\mathcal{L}_j(\rho) = \mathcal{D}_{L_j}(\rho)+\mathcal{D}_{H_G}(\rho)$ and $\mathcal{G}_{H_G}(\rho)$ where $L_j$ and $H_G$ are as defined above. The results of \cite{Johnson2016} imply that any convex combinations $\sum_j\alpha_j\mathcal{L}_j(\rho)$ makes GAS in mean, and thus satisfies Assumption \ref{assumption}.

\subsection{Numerical simulations}
    
    We will now present two simulations realized on a 5-qubit system that show the advantages of the proposed method. 
    
    Both simulations will have the same graph (depicted in Figure \ref{fig:sys_graph}), the same target state $\rho_G = U_G \ket{+}\bra{+}^{\otimes n}U_G^\dag$, and the same true initial condition $\rho_0 = U_G \ket{--++-}\bra{--++-}U_G^\dag$. The only difference between the two simulations will be the estimated initial condition $\rho_{0,e}$.
    
    \begin{figure}[h!]
		\centering
		\begin{tikzpicture}[scale = 0.4,line cap=round, line join=round]
		%\clip(-2.19,-2.49) rectangle (2.66,2.58);
		\draw[very thick,color=blue!40] (0,0) -- (4,0);
		\draw[very thick,color=blue!40] (-6,0) -- (0,0);
		\draw[very thick,color=blue!40] (4,0) -- (8,0);
		\draw[very thick,color=blue!40] (0,0) -- (-3,2);
		\draw[very thick,color=blue!40] (-6,0) -- (-3,2);
		
		\pgfmathsetmacro{\r}{1}
		\pgfmathsetmacro{\xc}{0}
		\pgfmathsetmacro{\yc}{0}
		\draw[color=orange,thick,fill=orange!30] (\xc,\yc) circle  (\r cm);
		\draw [dash pattern=on 3pt off 3pt,color=orange,thick] (\xc,\yc) ellipse (\r cm and 0.45*\r cm);
		\draw [->] (\xc,\yc)-- (\xc+\r/2*0.70,\yc+\r/2*1.07);
		\draw (\xc-1,\yc-1) node[anchor=north west] {$\ket{\psi_3}$};
		
		\pgfmathsetmacro{\xc}{4}
		\pgfmathsetmacro{\yc}{0}
		\draw[color=orange,thick,fill=orange!30] (\xc,\yc) circle  (\r cm);
		\draw [dash pattern=on 3pt off 3pt,color=orange,thick] (\xc,\yc) ellipse (\r cm and 0.45*\r cm);
		\draw [->] (\xc,\yc)-- (\xc+\r/2*0.70,\yc+\r/2*1.07);
		\draw (\xc-1,\yc-1) node[anchor=north west] {$\ket{\psi_4}$};
		
		\pgfmathsetmacro{\xc}{8}
		\pgfmathsetmacro{\yc}{0}
		\draw[color=orange,thick,fill=orange!30] (\xc,\yc) circle  (\r cm);
		\draw [dash pattern=on 3pt off 3pt,color=orange,thick] (\xc,\yc) ellipse (\r cm and 0.45*\r cm);
		\draw [->] (\xc,\yc)-- (\xc+\r/2*0.70,\yc+\r/2*1.07);
		\draw (\xc-1,\yc-1) node[anchor=north west] {$\ket{\psi_5}$};
		
		\pgfmathsetmacro{\xc}{-3}
		\pgfmathsetmacro{\yc}{2}
		\draw[color=orange,thick,fill=orange!30] (\xc,\yc) circle  (\r cm);
		\draw [dash pattern=on 3pt off 3pt,color=orange,thick] (\xc,\yc) ellipse (\r cm and 0.45*\r cm);
		\draw [->] (\xc,\yc)-- (\xc+\r/2*0.70,\yc+\r/2*1.07);
		\draw (\xc+1,\yc+1) node[anchor=north west] {$\ket{\psi_1}$};
		
		\pgfmathsetmacro{\xc}{-6}
		\pgfmathsetmacro{\yc}{0}
		\draw[color=orange,thick,fill=orange!30] (\xc,\yc) circle  (\r cm);
		\draw [dash pattern=on 3pt off 3pt,color=orange,thick] (\xc,\yc) ellipse (\r cm and 0.45*\r cm);
		\draw [->] (\xc,\yc)-- (\xc+\r/2*0.70,\yc+\r/2*1.07);
		\draw (\xc-1,\yc-1) node[anchor=north west] {$\ket{\psi_2}$};
		\end{tikzpicture}
		\caption{Graph configuration used in the simulations.}
		\label{fig:sys_graph}
	\end{figure}
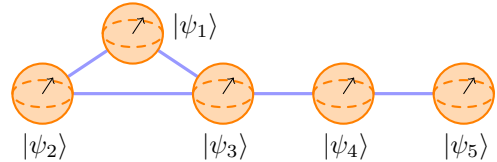

    In order to numerically compute the solution of the SME, we used the method proposed in \cite{Amini2011}, which we found the most reliable. Regarding the simulations of the cyclic and the state-based control laws, we simply used the Euler integration method. 
    
    The two simulations have been run with 
    step length $d_t=0.005,$ number of steps $N=1000,$ steps between switching $\epsilon=10,$ each with $1000$ realizations.
    The results of the two simulations will be shown as graphs of the trace norm distance of the true state from the target state $\frac{1}{2}\norm{\rho_t-\rho_G}_1$ against time. 
    
    \subsubsection{Simulation 1} 
    \label{par:sim_1}
    The first estimated initial state we consider is $\rho_{0,e} = 0.5[\rho_0^t+I_d/d]$. This introduces an additional  ``uniform'' uncertainty on the correct initialization of the filter. From Figure \ref{fig:results_1} we can observe that the measurement-based trajectory converges on average with a speed that is similar to the optimal one when the true initial state is available to both strategies. Both do improve convergence with respect to cyclic switching, which does not depend of course on the initial condition estimate.  
    For the wrong initial condition case, the measurement-based trajectory improves the convergence with respect to the state-based one, yet mostly in the central phase of the evolution. The measurement-based trajectories shown in Figure \ref{fig:results} are the average obtained from 1000 realizations. The true trajectories would resemble the dashed ones in Figure \ref{fig:case_detail}.  These results already highlight how the performance depends on the estimated state. This is further illustrated by the next set of simulations.
    
    \begin{figure}[h]
        \centering
        \includegraphics[scale=\myscale]{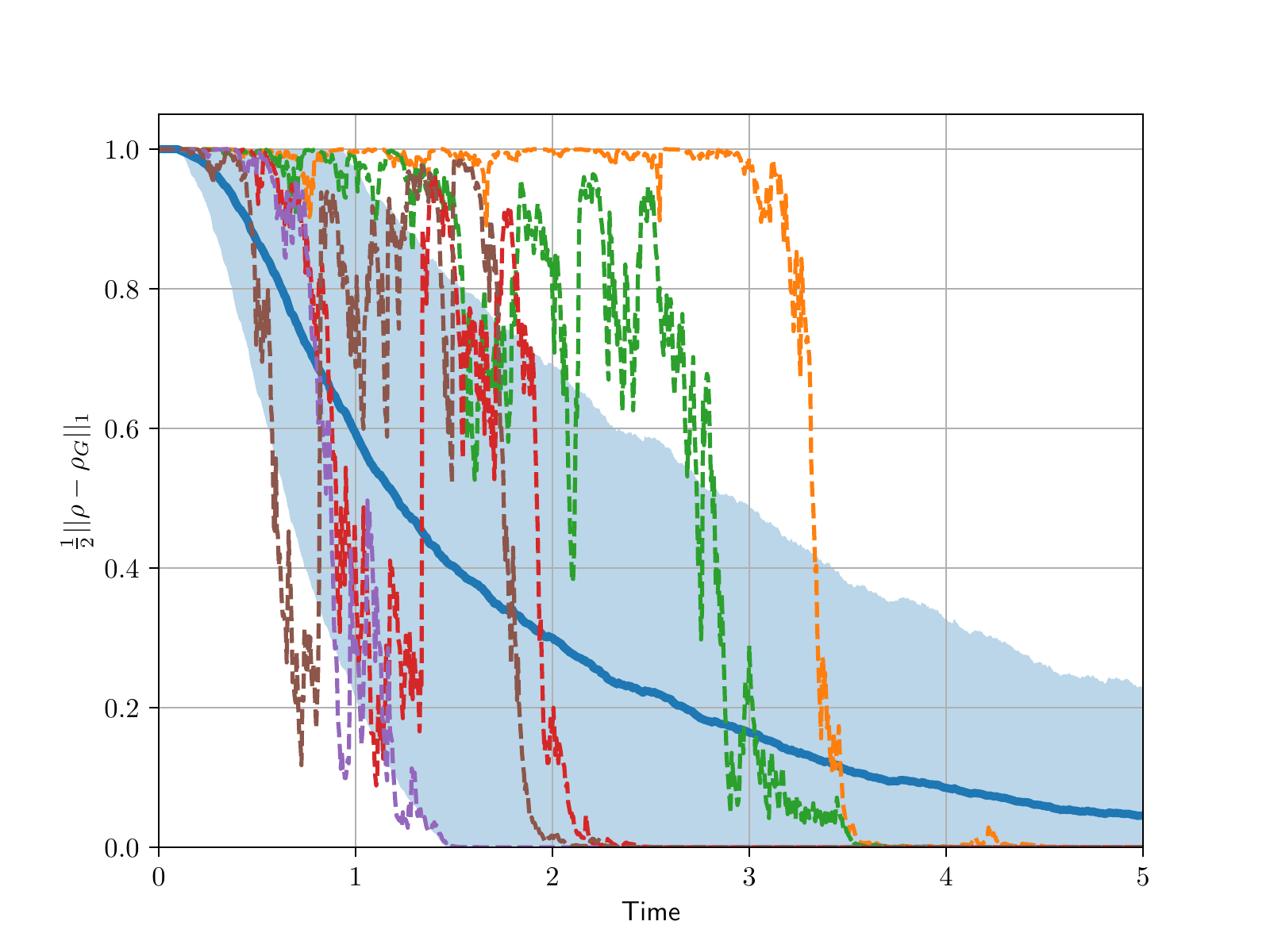}
        \caption{Evolution of the measurement-based trajectory in Simulation 1. In solid blue the average trajectory over 1000 realizations is shown, the light-blue area shows the average plus or minus one standard deviation while the dashed lines represent five typical realizations of the quantum trajectories.}
        \label{fig:case_detail}
    \end{figure}

    \subsubsection{Simulation 2}
    \label{par:sim_2}
    The second estimated initial state we consider is $\rho_{0,e} = 0.5[\rho_0^t +U_G \ket{++--+}\bra{++--+}U_G^\dag]$. This second case has been designed to highlight the strength of the proposed strategy. In fact, the state is a mixture of the true state and a state which is not only orthogonal to the true one but is obtained by flipping single-qubit states, so the marginal states are also orthogonal.  From the results shown in Figure \ref{fig:results_2} we can observe that this case is very different from the previous simulation. In particular, we can notice that correct initialization still provides the best convergence for both control laws. However, for the wrong initial condition case, the state-based trajectory has the same convergence rate as the cyclic trajectory, while the measurement based-trajectory provides a noticeable improvement, closer to the optimal performance. These results reinforce the intuition that closed-loop control, which can take advantage of the measurement outcome to obtain a better estimation of the state, can adapt in real-time and better avoid switching sequences that are not effective.

% \section{CONCLUSIONS}
%  In this work we explore the potential of feedback control {\em with dissipative control actions} for fast pure-state preparation,  proposing a measurement-based switching control law along with its proof of convergence in mean and {\blue almost surely}. 
%  We next present numerical simulations highlighting its advantage with respect to previously known open-loop techniques in the case of imperfect knowledge of the initial state. 

\printbibliography

\end{document}